\documentclass[smallextended,envcountsect]{svjour3}
\usepackage[margin=3cm]{geometry}
\usepackage{amscd, amsmath, amssymb, appendix, authblk, enumerate, enumitem, float, mathtools, natbib,  subcaption, tabu, tikz}
\usetikzlibrary{matrix,arrows,positioning,automata,shapes,shapes.multipart} 
\DeclarePairedDelimiter\floor{\lfloor}{\rfloor}

\newcommand{\ra}{\rightarrow}

\providecommand{\keywords}[1]
{
	\small	
	\textbf{\textit{Keywords---}} #1
}

\usepackage{color}

\smartqed

\title{On the comparison of incompatibility of split systems across different numbers of taxa}
\author{Michael Hendriksen \and Nils Kapust}
\institute{Institut f{\"u}r Molekulare Evolution, Heinrich-Heine Universit{\"a}t}

\begin{document}

\maketitle	
	
	\begin{abstract}
We consider the problem of how many phylogenetic trees it would take to display all splits in a given set, a problem related to $k$-compatibility. A set of trees that display every single possible split is termed a \textit{universal tree set}. In this note, we find the universal incompatibility $U(n)$, the minimal size of a universal tree set for $n$ taxa. By normalising incompatibility using $U(n)$, one can then compare incompatibility of split systems across different numbers of taxa. We demonstrate this application by comparing two SplitsTree networks derived from archaeal genomes, with different numbers of taxa.
	\end{abstract}


\section*{Acknowledgements}
The authors thanks Prof. Dr. W. F. Martin, the Volkswagen Foundation 93\_046 grant and the ERC Advanced Grant No. 666053 for their support during this research. The authors would also like to thank Andrew Francis for helpful comments on a draft, and Mike Steel, Fernando Tria and Falk Nagies for illuminating conversations on this topic. Finally, we thank the anonymous reviewers for their useful suggestions.

\section{Introduction}

Phylogenetic trees are ubiquitously used to represent the evolutionary history of organisms \citep{felsenstein2004inferring}. Each edge in an unrooted phylogenetic tree corresponds to a bipartition of the taxa set and a given phylogenetic tree can be uniquely identified with the set of bipartitions induced by its edge set. However, data can often produce conflicting results, whether through measurement error or complex biological phenomena such as incomplete lineage sorting or lateral gene transfer. This can result in splits that contradict each other.

This naturally gave rise to the concept of $k$-compatibility (first studied by this name by \cite{dress20012kn}, originally studied as $k$-cross-free families by \cite{inbook}), which, given a set of splits $S$, asks for the maximum size of a subset of $S$ in which every split is pairwise incompatible with one another. If this subset is of size $k$, the split system is termed $k$-compatible. We consider a related concept of incompatibility which is arguably more natural --- that of the smallest number of phylogenetic trees it would take to display all splits indicated by the data, a so-called \textit{minimal tree set}. In the case that $S$ is the set of all possible splits for a set of taxa $X$, we say that a set of phylogenetic trees that display all splits in $S$ is a \textit{universal tree set}.

In the present paper, we consider the question of maximum possible incompatibility in this way - that is, given a set of taxa $X$ of size $n \ge 2$, how large is a universal tree set of minimum size? We term this \textit{universal incompatibility}, and represent it with $U(n)$. This can also be characterised as finding the minimum $k \ge 1$ such that every split system $S$ on $X$ can be displayed by $k$ phylogenetic trees.

By characterising $U(n)$, one can then contextualise a split system in terms of how incompatible it is compared to the worst case scenario --- that is, the scenario in which every possible split is contained in our split system. Further, by normalising the minimal tree set size using $U(n)$, we can then compare incompatibility of split systems across different numbers of taxa.

Of particular interest is the fact that the widely-used SplitsTree software \citep{huson2006application} creates a so-called \textit{split network}, which is used to represent conflicting split signals from data. The present results will now allow those who use SplitsTree to fairly compare incompatibility of data across different numbers of taxa.

In Section 2 we provide background information. In Section 3 we prove some lemmas on bipartitions. In Section 4 we apply these lemmas and some classical theorems to prove the main result. In Section 5 we then apply these results to compare the incompatibility of two SplitsTree networks of different sizes derived from archaeal genomes.

\section{Background}

A phylogenetic tree on a set of taxa $X$ is a connected acyclic graph $(V,E)$ such that there are no vertices of degree $2$ and the degree-$1$ vertices (termed \textit{leaves}) are bijectively labelled by the elements of $X$.	

Recall that a \textit{split} $A|B$ of a set $X$ is a bipartition of $X$ into two non-empty sets
$A$, $B$; where $B = X \backslash A$. Define the \textit{size} of a split $A|B$ to be $\min(|A|,|B|)$. We denote by $\mathcal{S}(X)$ the set of
all splits of $X$, and any subset $S$ of $\mathcal{S}(X)$ is called a \textit{split system} on $X$. Any split of size $1$ is termed \textit{trivial}.

Given a phylogenetic tree $T=(V,E)$ on $X$, each edge can be associated with a split in the following way. If an edge $e$ is deleted from $T$, this disconnects the graph into two components, each with at least one labelled vertex. This naturally induces a bipartition on the leaf set $A|B$, which we call the split \textit{associated with} $e$, and we say that $A|B$ is \textit{displayed} by $T$ if it is associated with some edge of $T$.

It is well-known \citep{buneman1971recovery} that two splits $A|B$ and $C|D$ can only be displayed by the same phylogenetic tree if either one or two of the four intersections

\[A \cap C, A \cap D, B \cap C, B \cap D\]
is empty (noting that if two intersections are empty, $A|B$ and $C|D$ represent the same split). If this condition is met by each pair of splits in a split system $S$, we say that $S$ is \textit{pairwise compatible}, and if not, the set of splits is termed \textit{incompatible}. In fact, $S$ is pairwise compatible if and only if $S$ corresponds to a phylogenetic tree in the following way.

\begin{theorem}[Splits Equivalence Theorem, \cite{buneman1971recovery}]
	Let $S$ be a collection of splits on $X$. Then,	$S$ is the set of splits of some phylogenetic $X$-tree $T$ if and only if $S$ contains all trivial splits on $X$ and $S$ is pairwise compatible. The tree $T$ is unique up to isomorphism.
\end{theorem}

We will therefore henceforth consider a phylogenetic tree and the corresponding pairwise compatible split set as interchangeable.

In a biological context, sets of incompatible splits frequently arise from data, and biologists wish to quantify the extent to which the set is incompatible. This naturally gave rise to the definition of $k$-compatibility. We say that a split system is \textit{$k$-compatible} if it does not contain a subset of $k+1$ pairwise incompatible splits (for $k \ge 1$). A related concept is that of a \textit{minimal tree set} for a given split system. Given a split system $S$, we say it has a \textit{tree set of size $k$} if one can find a set of $k \ge 1$ phylogenetic trees $\mathcal{T}$ on the same set $X$ such that every split in $S$ is displayed by at least one phylogenetic tree in $\mathcal{T}$. In the case that $S=\mathcal{S}(X)$ we say that $\mathcal{T}$ is a \textit{universal tree set}. Define the function $U(n)$ to be the value of $|\mathcal{T}|$, where $\mathcal{T}$ is a minimal universal tree set on a set of taxa of size $n \ge 2$.

We say that a set of $k$ phylogenetic trees $\mathcal{T}$ that displays every split in a set of splits $S$ is \textit{minimal} with respect to $S$ if there are no sets of $k-1$ phylogenetic trees with this property. An example of a minimal universal tree set for $5$ leaves is shown in Figure \ref{f:MU5}, showing that $U(5) \le 5$. Indeed, due to the fact that each tree on $5$ leaves can display at most two splits, and there are ten unique non-trivial splits on $5$ leaves, we conclude $U(5) \ge 5$ and thus $U(5)=5$.

We note here that $k$-compatibility of a split system and minimal tree set size of a split system are related concepts. Certainly if a split system is $k$-compatible, the minimal tree set size is at least $k$, as given a set of $k$ pairwise incompatible splits each must be displayed by a different phylogenetic tree. Therefore minimal tree set size is bounded below by $k$-compatibility, but they are not the same, as the following example shows. We thank an anonymous reviewer for this example.

\begin{example}
	Let $X= \{1,2,3,4,5\}$ and let $S = \{12|345,23|145,34|125,45|123, 15|234\}$. Then, for instance, $12|345$ and $23|145$ are incompatible, but it is impossible to find a pairwise incompatible subset of size $3$, so $S$ is $2$-compatible. However, as a phylogenetic tree with $5$ leaves can display at most $2$ splits of size $2$, the minimal tree set size of $S$ is $3$.
\end{example}

Finally, we note as an aside that a universal tree set (minimal or otherwise) has no requirement that all phylogenetic trees in the set must be binary. However, given a minimal universal tree set containing a strictly non-binary phylogenetic tree $T$, one can replace $T$ with a binary refinement of $T$ without compromising minimality of the set --- the tree set will still display all splits on $X$. Hence, for a given set $X$ there will always exist a minimal universal tree set on $X$ consisting only of binary phylogenetic trees. Such a minimal universal tree set will generally not be unique, but this will not affect the calculation of $U(n)$.

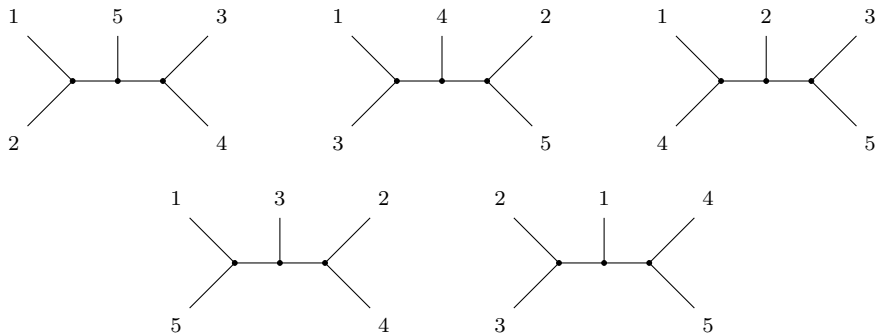
\begin{figure}[h]
	\centering
	\begin{tikzpicture}[scale=0.6]
	\draw (0,0)--(1,1)--(0,2);
	\draw (1,1)--(3,1)--(4,2);
	\draw (3,1)--(4,0);
	\draw (2,1)--(2,2);

	\draw[fill] (1,1) circle [radius=1.5pt];
	\draw[fill] (2,1) circle [radius=1.5pt];
	\draw[fill] (3,1) circle [radius=1.5pt];
	
	\node[below left] at (0,0) {$\mathstrut 2$};
	\node[above left] at (0,2) {$\mathstrut 1$};
	\node[above] at (2,2) {$\mathstrut 5$};
	\node[below right] at (4,0) {$\mathstrut 4$};
	\node[above right] at (4,2) {$\mathstrut 3$};
	\end{tikzpicture}
	\hspace{1cm}
	\begin{tikzpicture}[scale=0.6]
	\draw (0,0)--(1,1)--(0,2);
	\draw (1,1)--(3,1)--(4,2);
	\draw (3,1)--(4,0);
	\draw (2,1)--(2,2);

	\draw[fill] (1,1) circle [radius=1.5pt];
	\draw[fill] (2,1) circle [radius=1.5pt];
	\draw[fill] (3,1) circle [radius=1.5pt];
	
	\node[below left] at (0,0) {$\mathstrut 3$};
	\node[above left] at (0,2) {$\mathstrut 1$};
	\node[above] at (2,2) {$\mathstrut 4$};
	\node[below right] at (4,0) {$\mathstrut 5$};
	\node[above right] at (4,2) {$\mathstrut 2$};
	\end{tikzpicture}
	\hspace{1cm}
	\begin{tikzpicture}[scale=0.6]
	\draw (0,0)--(1,1)--(0,2);
	\draw (1,1)--(3,1)--(4,2);
	\draw (3,1)--(4,0);
	\draw (2,1)--(2,2);

	\draw[fill] (1,1) circle [radius=1.5pt];
	\draw[fill] (2,1) circle [radius=1.5pt];
	\draw[fill] (3,1) circle [radius=1.5pt];
	
	\node[below left] at (0,0) {$\mathstrut 4$};
	\node[above left] at (0,2) {$\mathstrut 1$};
	\node[above] at (2,2) {$\mathstrut 2$};
	\node[below right] at (4,0) {$\mathstrut 5$};
	\node[above right] at (4,2) {$\mathstrut 3$};
	\end{tikzpicture}
	\par\medskip
	\begin{tikzpicture}[scale=0.6]
	\draw (0,0)--(1,1)--(0,2);
	\draw (1,1)--(3,1)--(4,2);
	\draw (3,1)--(4,0);
	\draw (2,1)--(2,2);

	\draw[fill] (1,1) circle [radius=1.5pt];
	\draw[fill] (2,1) circle [radius=1.5pt];
	\draw[fill] (3,1) circle [radius=1.5pt];
	
	\node[below left] at (0,0) {$\mathstrut 5$};
	\node[above left] at (0,2) {$\mathstrut 1$};
	\node[above] at (2,2) {$\mathstrut 3$};
	\node[below right] at (4,0) {$\mathstrut 4$};
	\node[above right] at (4,2) {$\mathstrut 2$};
	\end{tikzpicture}
	\hspace{1cm}
	\begin{tikzpicture}[scale=0.6]
	\draw (0,0)--(1,1)--(0,2);
	\draw (1,1)--(3,1)--(4,2);
	\draw (3,1)--(4,0);
	\draw (2,1)--(2,2);
	
	\draw[fill] (1,1) circle [radius=1.5pt];
	\draw[fill] (2,1) circle [radius=1.5pt];
	\draw[fill] (3,1) circle [radius=1.5pt];
	
	\node[below left] at (0,0) {$\mathstrut 3$};
	\node[above left] at (0,2) {$\mathstrut 2$};
	\node[above] at (2,2) {$\mathstrut 1$};
	\node[below right] at (4,0) {$\mathstrut 5$};
	\node[above right] at (4,2) {$\mathstrut 4$};
	\end{tikzpicture}
	\caption{A minimal universal tree set on the $5$-leaf set $\{ 1,2,3,4,5 \}$.}
	\label{f:MU5}
\end{figure}

\section{Combinatorial Results on Bipartitions}

In order to discover how many trees are required to display all of the splits in a set, we will first present some results on the maximum number of splits of the largest possible size that can be displayed by a given tree. We will address this question for even and odd $n$ separately. 

\begin{lemma}
	\label{l:EvenSplits}
	Let $n$ be an even integer. Then a phylogenetic tree on $n$ leaves has at most one split of size $n/2$.
\end{lemma}

\begin{proof}
	Let $A|B$ and $C|D$ be a pair of splits displayed by the same phylogenetic tree, of size $n/2$. Then one of $A \cap C, A \cap D, B \cap C$ or $B \cap D$ is empty. Without loss of generality, suppose that $A \cap C$ is empty. Then $A \subseteq (X \backslash C) = D$, but since both partitions are of size $n/2$, it follows that $A=D$, so $B=C$ and thus $A|B$ and $C|D$ are equivalent partitions. Combined with the Splits Equivalence Theorem, the lemma follows. \qed

\end{proof}

Note that of course a phylogenetic tree need not have any such split, as we can consider the star tree --- that is, the tree with only trivial splits --- for any number of leaves $n \ge 4$. This Lemma gives a lower bound for $U(n)$ for even $n$, as phylogenetic trees with an even number of leaves can have at most one split of size $m=\frac{n}{2}$, and $\frac{1}{2} \binom{n}{m}$ is the number of such splits for a given $n$. In fact, $U(n)$ actually equals this lower bound in the even case, as we will see in Theorem \ref{t:MinEven}. 

\begin{lemma}
	\label{l:OddSplits}
	Let $n=2m+1$ where $m \ge 2$ is a positive integer. Then a phylogenetic tree on $n$ leaves displays at most two splits of size $m$.
\end{lemma}

\begin{proof}
	Seeking a contradiction, let $A|B$, $C|D$ and $E|F$ be three distinct splits on the same phylogenetic tree, so that $|A|=|C|=|E|=m,|B|=|D|=|F|=m+1$. Then one of $A \cap C, A \cap D, B \cap C$ or $B \cap D$ is empty. As $A|B$ and $C|D$ are distinct splits, it must be the case that $A \cap C$ is empty. Therefore, $A \subset D$ and $C \subset B$; in fact, $D = A \cup \{x\}$ for some taxon $x$. This implies that $C=B \backslash \{x\}$.
	
	By similar logic, $F = A \cup \{y\}$ and $E=B \backslash \{y\}$ for some taxon $y$ so that $y \ne x$ (since $E|F$ and $C|D$ are distinct). But then $C \cap E, C \cap F, D \cap E$ and $D \cap F$ must all be non-empty, which is a contradiction. Combined with the Splits Equivalence Theorem, the lemma follows. \qed
\end{proof}

One can observe that in the $n=3$ case there are $3$ such splits (and that a minimal universal tree set on $3$ leaves consists of just the star tree on $3$ leaves). Outside of this case, as each phylogenetic tree with an odd number of leaves can display up to two of these splits, a natural follow-up question is whether there are any obstructions to pairing all such splits in this way. That is, can we partition the largest splits into compatible pairs so that each tree in our set displays a unique pair of splits of size $m$?

Fortunately we (almost) can, using the concept of matchings. We will need two definitions before we can see this.

\begin{definition}
	A \textit{matching} $M$ of a graph $G$ is a set of edges of $G$ such that no two edges share the same vertex. A \textit{defect}-$d$ matching $M$ is a matching so that all except $d$ vertices of $G$ have an incident edge from $M$. Defect-$0$ matchings are also referred to as \textit{perfect} matchings.
\end{definition}

Let $m<n$ and $Bip(n,m)$ be the set of bipartitions of $n$ of size $m$. We intend to construct a graph in which the vertices are the elements of $Bip(n,m)$ and there is an edge between two vertices if and only if they are compatible, distinct bipartitions. We will then find matchings on this graph, with the aim of having as small a defect as possible - thus pairing our large splits as efficiently as possible.

This will require a graph theoretic result courtesy of Little, Grant and Holton, which itself requires an additional definition.

\begin{definition}
	A graph $G$ is said to be \textit{vertex-transitive} if, given any two vertices $v_1$ and $v_2$ of $G$, there is some automorphism $f\colon V(G)\rightarrow V(G)$ such that	$f(v_{1})=v_{2}$ and $f(v_{2})=v_{1}$. 
\end{definition}

We can now state the following theorem.

\begin{theorem}[\cite{little1975defect}]
	\label{t:PerfectMatch}
	Every connected vertex-transitive graph with an even number of vertices has a perfect matching  and every connected vertex-transitive graph with an odd number of vertices has a defect-$1$ matching.
\end{theorem}

Hence it suffices to show that our graph $G$ is connected and vertex-transitive when $n=2m+1$. However, Theorem \ref{t:PerfectMatch} differentiates between graphs with even and odd numbers of vertices, and the number of vertices of $G$ is $\binom{n}{m}$, which can be even or odd depending on the values of $m$ and $n$. Therefore we will need to distinguish these cases. This will require a short detour on the properties of the binomial coefficient, using a theorem of Kummer.

Let $a,b$ be positive integers, and assume that $a+b$ has $r$ digits in base $p$. Then we can assume $a + b, a,$ and $b$ all have $r$ digits in base $p$ by adding leading $0$’s if necessary. Denote the $i$-th digit from the right of $a$ and $b$ by $a_i$ and $b_i$ respectively. We define $\phi(1)=0$ if $a_1 + b_1 < p$, and $\phi(1)=1$ otherwise. Then, for $2 \le i \le r$, we define $\phi(i)=0$ if $a_i+b_i + \phi(i-1) < p$ and $\phi(i)=1$ otherwise. Then the \textit{number of carries when adding $a$ and $b$ in base $p$} is the sum

\[\sum_{i=1}^{r} \phi(i).\] 

\begin{example}
	Let $a=15, b=4$, so in base $2$ we have that $a$ is $01111$, $b$ is $00100$ and $a+b$ is $10011$. Then $\phi(1) = 0$ as $1+0 < 2$, $\phi(2)=0$ as $1+0+0 <2$, and similarly, $\phi(3)=1$, $\phi(4)=1$ and $\phi(5)=0$. Hence the number of carries when adding $15$ and $4$ in base $2$ is $0+0+1+1+0=2$.
\end{example}

We can now state Kummer's Theorem.

\begin{theorem}[\cite{kummer1852erganzungssatze}]
	If $p$ is a prime, then the largest power of $p$ that divides $\binom{m+n}{n}$, for $m$ and $n$ non-negative, is the number of carries when $m$ and $n-m$ are added in base $p$.
\end{theorem}

As a simple corollary to this, by considering base $2$ we get

\begin{corollary}
	\label{c:poweroftwo}
	If $n = 2m+1$ for positive integer $m$, then $\binom{n}{m}$ is odd if and only if $n$ is one less than a power of $2$.
\end{corollary}

For example, $\binom{7}{3} = 35$ and $\binom{15}{7} = 6435$ correspond to the cases $m=3$ and $m=7$ respectively, and yield odd results as $n$ is one less than a power of $2$.

We can now state our lemma on pairing partitions, which will use the concept of matchings in the proof.

\begin{lemma}
	\label{OddMatching}
	Let $n = 2m+1$, where $m \ge 2$. If $n=2^k-1$ for some integer $k$ then we can partition $Bip(n,m)$ into compatible pairs $\{A|B,C|D\}$with one leftover, otherwise we can partition $Bip(n,m)$ into compatible pairs with no leftovers.
\end{lemma}

\begin{proof}
	Let $G$ be the graph whose vertices are the elements of $Bip(n,m)$, and there is an edge between vertices $A|B$ and $C|D$ if and only if they are compatible bipartitions. We first show that $G$ is connected. 
	
	Let $A|B$ and $C|D$ be a pair of splits of size $m$ and let $Int(A|B,C|D) =k$ be the size of the intersection between the partition of $A|B$ with $m+1$ elements and the partition of $C|D$ with $m+1$ elements - without loss of generality supposing they are $A$ and $C$ respectively. Additionally, note $k>1$. We claim that either $A|B$ and $C|D$ coincide (which occurs if and only if $k=m+1$) or there exists a split $E|F$ in the same connected component of $G$ as $A|B$ so that $Int(E|F,C|D)>k$. As $m+1$ is finite, this implies that $A|B$ and $C|D$ must be in the same connected component, and since $A|B$ and $C|D$ were arbitrary, that $G$ is connected. It therefore remains to be shown that there exists a split $E|F$ in the same connected component of $G$ as $A|B$ so that $Int(E|F,C|D)>k$.
	
	Suppose $A|B$ and $C|D$ do not coincide and consider the split $A'|B'$ obtained by taking some element $x \in (A \backslash C)$ (which must exist as both $A$ and $C$ have size $m+1$ and do not coincide), and letting $A'=A \backslash \{x\}$ and $B'=B \cup \{x\}$. Note that in this case, $B'$ is now the partition of size $m+1$, and further that $A|B$ and $A'|B'$ are compatible as $B \cap A'$ is empty. Hence $A|B$ and $A'|B'$ are in the same connected component of $G$ (indeed, there is an edge between them).
	
	If $Int(A'|B',C|D)>k$, then the claim is proven by taking $E|F=A'|B'$. Otherwise, $Int(A'|B',C|D)= k' \le k < m+1$, and so $A'|B'$ and $C|D$ do not coincide. Of course, $k' > 0$ as both $B'$ and $C$ again have size $m+1$ and do not coincide. We can therefore take some element $y \in B \cap C$ and form the new split $A''|B''$ (also of size $m$), where $A'' = A' \cup \{y \}$ and $B'' = B' \backslash \{y\}$. Note that $A'|B'$ and $A''|B''$ are compatible as $A' \cap B''$ is empty, and therefore $A|B,A'|B'$ and $A''|B''$ are all in the same connected component.
	
	Then $Int(A''|B'',C|D)=k+1$ as to form $A''$ we removed an element from $A$ that was not in $C$, and then added an element that was in $C$. By taking $E|F = A''|B''$ the claim is therefore proven. Hence $G$ is connected.

	We will now show that $G$ is vertex-transitive. 
	
	Let $A|B$ and $C|D$ be two vertices of $G$, where $|A|=|C|=m, |B|=|D|=m+1$. Let $\sigma$ be a permutation of $X$ so that $\sigma$ applied to each taxon in $A$ obtains $C$, and similarly applied to $B$ obtains $D$. Then the induced action by applying this permutation to every bipartition in $G$ is an automorphism that maps $A|B$ to $C|D$, hence $G$ is vertex-transitive.
	
	By Theorem \ref{t:PerfectMatch}, as $G$ is connected and vertex-transitive, if $G$ has an even number of vertices there exists a perfect matching  and if $G$ has an odd number of vertices there exists a defect-$1$ matching. As there is an edge between vertices $A|B$ and $C|D$ if and only if they are compatible bipartitions, the existence of a perfect matching is equivalent to the existence of a partitioning of the vertices into compatible pairs of splits. Similarly, the existence of a defect-$1$ matching is equivalent to the existence of a partitioning of the vertices into compatible pairs of splits with one leftover split.

	Finally, the number of vertices of $G$ is $\binom{n}{m}$, and Corollary \ref{c:poweroftwo} of Kummer's Theorem tells us exactly when this value is odd and when it is even. The claim follows. \qed
\end{proof}
	
\begin{example}
	Let $X = \{1,2,3,4,5\}$, so $n=5$ and $m=2$. Then we can partition $Bip(5,2)$ into compatible pairs, for example $\{12|345,34|125\}$,$\{13|245,25|134\}$,$\{14|235,35|124\}$,$\{15|234,24|135\}$ and $\{23|145,45|123\}$. Note that each pair is precisely the set of non-trivial bipartitions corresponding to a unique tree from Figure \ref{f:MU5}.
\end{example}

\section{Minimal Universal Tree Sets}

We will shortly prove the main theorem of this paper, Theorem \ref{t:MinEven}, pending some useful theorems. The statement of the theorem requires the following definition.

\begin{definition}
	Let $x$ be a real number. Then the ceiling of $x$, denoted by $\left\lceil x \right\rceil$ is the smallest integer $n$ so that $n \ge x$. The floor of $x$, denoted by  $\floor{x}$, is the largest integer $n$ so that $n \le x$.
\end{definition}

\begin{theorem}[Main Theorem]
	\label{t:MinEven}
	Let $X$ be a set of size $n \ge 2$, and let $m$ be a positive integer such that $n=2m$ if $n$ is even and $n=2m+1$ if $n$ is odd. Then a minimal universal tree set for $X$ has size 
	
	\[\displaystyle  U(n) = \left\lceil\frac{1}{2} \binom{n}{m} \right\rceil. \] 
\end{theorem}

To prove this, we will need a few useful classical theorems from extremal set theory. A \textit{poset} is a set $P$ together with a binary relation $\le$ on its elements that is reflexive ($x \le x$ for all $x \in P$), antisymmetric (if $x \le y$ and $y \le x$ then $x=y$ for all pairs $x,y \in P$) and transitive (if $x \le y$ and $y \le z$ then $x \le z$ for all triples $x,y,z \in P$). Define $P(X)$ to be the poset on the power set of $X$, ordered by set inclusion. In particular we need a theorem of Sperner and a theorem of Dilworth, which we will use to partition $P(X)$ into chains, which is necessary for constructing sets of compatible splits.

\begin{definition}
	Let $P$ be a poset with a reflexive, antisymmetric and transitive binary relation $\le$ on its elements. Two elements $x$ and $y$ of $P$ are said to be \textit{comparable} if either $x \le y$ or $y \le x$. We call a subset $S$ of $P$ a \textit{chain} if any two of its elements are comparable, and an \textit{antichain} if no distinct pair of its elements are comparable.
\end{definition}

Note that for our example, $P(X)$, the power set of $X$ with the binary relation of set inclusion, a chain is a set $S$ of sets in $P(X)$ so that for any pair of sets $A,B$ in $S$, either $A$ is contained in $B$ or $B$ is contained in $A$. An antichain in our example is a set $S$ of sets in $P(X)$ so that for any pair of distinct sets $A$ and $B$, neither is contained in the other.

\begin{theorem}[\cite{sperner1928satz}]
	Let $X$ be a set of size $n$. Then the largest antichain in $P(X)$ has size $\binom{n}{\floor{n/2}}$.
\end{theorem}

\begin{theorem}[\cite{dilworth1950decomposition}]
	Let $P$ be a poset and suppose the largest antichain in $P$ has size $r$. Then $P$ can be partitioned into $r$ chains.
\end{theorem}

We are now ready for the proof, which we will divide into even and odd cases.

\begin{proof}[Proof for Theorem \ref{t:MinEven} when $n$ is even]
	Let $n=2m$ and consider the poset $P(X)$. Sperner's Theorem states that the largest antichain of $P(X)$ has size $\binom{n}{m}$, and Dilworth's Theorem implies that we can therefore partition $P(X)$ into $\binom{n}{m}$ chains. Certainly no subset of size $m$ can be contained in a distinct subset of size $m$, so each chain contains at most one subset of size $m$. As there are exactly $\binom{n}{m}$ such subsets, each chain must therefore contain exactly one subset of size $m$.
	
	Select any such partition into chains and consider the graph $Chain(X)$ in which the vertices are the non-empty subsets in $P(X)$ of size $m$ or less, and there is an edge $e=(U,V)$  if and only if
	
	\begin{enumerate} 
		\item $U \le V$ in $P(X)$ and there is no set $W \in P(X)$ where $W \ne U,V$ such that $U \le W \le V$; and 
		\item $U$ and $V$ are elements of the same chain.
	\end{enumerate}

	Let 
	\[\gamma: V(Chain(X)) \ra \mathcal{S}(X)\] 
	be the function that maps the subset $A$ to the bipartition $A|(X \backslash A)$. Note that $\gamma$ is not injective, and $\gamma(A)=\gamma(B)$ if and only if $A=X-B$. Finally, define $BipChain(X)$ to be the graph consisting of vertices $\gamma(V(Chain(X)))$ and an edge $e=(\gamma(U),\gamma(V))$ if and only if $(U,V) \in E(Chain(X))$.
	
	In particular, as $\gamma(A)=\gamma(B)$ if and only if $A=X-B$, which occurs precisely when $|A|=m$. Thus $BipChain(X)$ has exactly half the number of components that $Chain(X)$ does, as for each pair $A,B$ such that $A=X-B$, the chain in $Chain(X)$ containing $A$ and the chain in $Chain(X)$ containing $B$ are mapped by $\gamma$ to the same component in $BipChain(X)$.
	
	Hence, the number of components of $\gamma(Chain(X))$ will be \[\displaystyle k = \frac{1}{2} \binom{n}{m}. \] 
	
	We construct a universal tree set as follows. Denote the components of $\gamma(Chain(X))$ by $C_1,...,C_k$.

	We claim that the set of phylogenetic trees corresponding to the sets of bipartitions $V(C_1),...,V(C_k)$ via the Splits Equivalence theorem, is a universal tree set, in particular that all $V(C_i)$ are sets of compatible bipartitions.
	
	First let the unique bipartition of size $m$ in $V(C_i)$ be $A|B$. Suppose we have two distinct bipartitions, $C|D$ and $C'|D'$. In order to show compatibility, it suffices (but is not necessary) to show that one of $C$ or $D$ is contained in one of $C'$ or $D'$, or the reverse, as the inclusion requirement quickly holds.
	
	If $C'|D' = A|B$, then certainly $C$ or $D$ is a subset of $A$ or $B$ by the chain construction. We therefore assume neither bipartition is $A|B$, and without loss of generality that $|C| < |D|$ and $|C'| < |D'|$. Then,  either $C$ and $C'$ are subsets of the same set ($A$ or $B$), or one is a subset of $A$ and the other of $B$. If $C$ and $C'$ are subsets of the same set, then $C$ and $C'$ are from the same component of $Chain(X)$ and therefore $C \subseteq C'$ or the reverse. However, if they are subsets of different sets, then $C \subset X \backslash C' = D'$. Therefore in all cases $V(C_i)$ is a set of compatible bipartitions, and as every bipartition is present in some $V(C_k)$ the tree set corresponding to $V(C_1),...,V(C_k)$ via the Splits Equivalence theorem is a universal tree set. 
	
	It finally remains to confirm that 
	\[\displaystyle k = \left\lceil\frac{1}{2} \binom{n}{m} \right\rceil. \]  
	is the minimum possible value. However, as each phylogenetic tree in any universal tree set can contain at most one split of size $m$ by Lemma \ref{l:EvenSplits} (of which there are $\binom{n}{m}/2$, which is equal to the desired formula when $n$ is even), the set is minimal. The theorem follows. \qed
\end{proof}

The proof for odd $n$ proceeds similarly, with some small modifications.

\begin{proof}[Proof for Theorem \ref{t:MinEven} when $n$ is odd]
	Let $n=2m+1$ and consider the poset $P(X)$. Sperner's Theorem states that the largest antichain of $P(X)$ has size $\binom{n}{m}$, and Dilworth's Theorem implies that we can therefore partition $P(X)$ into $\binom{n}{m}$ chains, each of which contains one subset of size $m$, by a similar counting argument to the even case.
	
	Select any such partition into chains and consider the graph $Chain(X)$ in which the vertices are the non-empty subsets in $P(X)$ of size $m$ or less, and there is an edge $e=(U,V)$  if and only if

\begin{enumerate} 
	\item $U \le V$ in $P(X)$ and there is no set $W \in P(X)$ where $W \ne U,V$ such that $U \le W \le V$; and
	\item $U$ and $V$ are elements of the same chain.
\end{enumerate}

	Let 
	\[\gamma: V(Chain(X)) \ra \mathcal{S}(X)\] 
	be the function that maps the subset $A$ to the bipartition $A|(X \backslash A)$. Note that in this case $\gamma$ \textit{is} injective, as $X-A$ must have a size larger than $m$.
	
	 We now perform one additional modification. By Lemma \ref{OddMatching}, there exists a matching $M$ between the bipartitions of the form $A|B$ where $|A|=m, |B|=m+1$ with at most one unpaired bipartition. For each of the
	 	 \[\displaystyle \left\lfloor\frac{1}{2} \binom{n}{m} \right\rfloor  \]
	  such matchings $(A|B,C|D)$ in $M$, add this edge to $\gamma(Chain(X))$, and call the resulting graph $BipChain(X)$.
	 
	 As $\gamma(Chain(X))$ contained $\binom{n}{m}$ connected components and each additional edge reduced the number of components by one, the number of connected components of $BipChain(X)$ will therefore be 
	 \[\displaystyle k = \binom{n}{m} - \left\lfloor\frac{1}{2} \binom{n}{m} \right\rfloor = \left\lceil\frac{1}{2} \binom{n}{m} \right\rceil . \]
	 
	 We construct a universal tree set as follows. Denote the components of $\gamma(Chain(X))$ by $C_1,...,C_k$.
	 
	 We claim that the set of phylogenetic trees corresponding to the sets of bipartitions $V(C_1),...,V(C_k)$ via the Splits Equivalence theorem, is a universal tree set, in particular that all $V(C_i)$ are sets of compatible bipartitions.
	 
	 Let the two (distinct) splits of size $m$ in $V(C_i)$ be $A|B$ and $A'|B'$, and let $C|D$ and $C'|D'$ be any two distinct bipartitions in $V(C_i)$. If $C|D$ and $C'|D'$ are in the same component in $Chain(X)$, the options proceed analogously to the even case, and $C|D$ is compatible with $C'|D'$. 
	 
	 Therefore instead suppose that $C|D$ was in the same component of $Chain(X)$ as $A|B$ and $C'|D'$ was in the same component of $Chain(X)$ as $A'|B'$. In particular, without loss of generality suppose that $C \subseteq A$ and $C' \subseteq A'$.
	 
	 As $A|B$ and $A'|B'$ are compatible, one of $A \cap A', A \cap B', B \cap A'$ or $B \cap B'$ are empty, but as $B$ and $B'$ have size $m+1$ the only possibility is that $A \cap A'$ is empty. Then as $C \subseteq A$ and $C' \subseteq A'$ it follows that $C \cap C'$ is empty, so $C|D$ and $C'|D'$ are compatible. 
	 
	 It finally remains to confirm that 
	 \[\displaystyle k = \left\lceil\frac{1}{2} \binom{n}{m} \right\rceil. \]  
	 is the minimum possible value. By Lemma \ref{l:OddSplits} in each phylogenetic tree in a universal tree set there can be at most two splits of size $m$, of which there are $\binom{n}{m}$, implying that 	 
	 \[\displaystyle k = \left\lceil\frac{1}{2} \binom{n}{m} \right\rceil \]  
	 is a lower bound (noting that the binomial coefficient can be odd, hence the ceiling function). \qed
\end{proof}

We note here that for $U(n)$, the associated integer sequence  ($3,5,10,18,35,63...$) appears in the OEIS as sequence A002661 for $n \ge 4$ \citep{OEISseq}.

\section{Applications and Discussion}

A split network is a combinatorial generalisation of a tree in which sets of edges are associated with bipartitions instead of just a single edge \citep{huson2006application}. Due to this, each split network is associated with a set of splits (which are not necessarily compatible) and so give us a rich source of split systems for which we can find minimal tree sets.

We now consider two SplitsTree-generated \citep{huson2006application} split networks derived from archaeal genomes and depicted in Figure \ref{f:archaea}. The first network $N_1$ contains 13 taxa, and the second network, $N_2$ is the network obtained after removal of the single taxon \textit{Methanococcus maripaludis}, leaving $12$ taxa.

To generate these networks, $39$ universal archaeal protein families gathered by \cite{nelson2015origins}, were used for a BLAST search against archaeal genomes obtained from the RefSeq 2016 database \citep{o2016reference} with an identity threshold of 20\% and an e-value cut-off of $10^{-5}$. For each best hit, alignments were generated using MAFFT v7.299b (linsi) \citep{katoh2002mafft} and concatenated using an in-house Python script. These concatenated alignments were used to draw a Neighbor-Net using SplitsTree4. We note here that all splits represented by the data are also present in the networks shown in Figure \ref{f:archaea}. The source files are available as supplementary data.

Using a short Python program that we have made available online \citep{Hendriksen2020}, we analysed the splits corresponding to each network and found a set of $4$ phylogenetic trees that display all splits in $N_1$, and a set of $3$ phylogenetic trees that display all splits in $N_2$. We note that the splits were analysed from the source file, so necessarily included all splits indicated by the data. These were then shown to be minimal by hand - both networks display $3$ incompatible splits of size $6$, and $N_1$ additionally displays a partition $A|B$ of size $5$ that is incompatible with each of the first $3$. Therefore, if we denote the minimal tree set size of the split system associated with a network $N$ by $\kappa(N)$, we know $\kappa(N_1)=4$ and $\kappa(N_2)=3$. It was observed that there are several possible minimal tree sets that can be computed for the set of splits corresponding to $N_1$ and $N_2$, but this of course does not affect the values of $\kappa(N_1)$ and $\kappa(N_2)$.

If we denote the number of leaves of a network $N$ by $|N|$, then we can define the \textit{normalised tree set size} $Norm(N)$ to be 

\[Norm(N)=\frac{\kappa(N)}{U(|N|)}.\]

Now, as $U(13)=858$ and $U(12)=462$, we can normalise these minimal tree set sizes, and as

\[Norm(N_1) = \frac{4}{858} < Norm(N_2)= \frac{3}{462}, \]
from the perspective relative to universal incompatibility, $N_2$ is `more incompatible' than $N_1$.

Although the underlying dataset consists of Archaeal proteins, which are known to have lateral gene transfer events \citep{nelson2015origins} the specific proteins which are used here are mainly ribosomal proteins. Ribosomal subunits are involved in the cellular process of translation. It is known that they are very conserved proteins across all life forms \citep{ban2014new}. We deliberately made the choice to examine Archaea given the lateral gene transfer events, so that we could expect some discordance, but selected ribosomal proteins to limit the extent thereof for ease of analysis.

The organism which is removed from $N_1$ to obtain $N_2$, \textit{Methanococcus maripaludi}, is a fully-sequenced model organism among hydrogenotrophic methanogens \citep{goyal2016metabolic}, and is the only member of the genus \textit{Methanococcus} in our dataset.

As the only organism of the genus \textit{Methanococcus}, we would expect the \textit{Methanococcus} to be evolutionarily distinct from the remaining organisms. We would therefore predict that it would contribute proportionately less to the incompatibility of the data with respect to the remaining twelve organisms of the dataset, so the result that $N_2$ is relatively more incompatible than $N_1$ is as anticipated.

Mathematically speaking, there are several natural extensions to the problem of minimal universal tree sets for future research. For instance, one avenue could be to investigate how the value changes if we instead ask for a minimal universal set of networks with at most $k$ reticulations. The probabilistic analogue of the question would also be interesting - how likely is it, given a set of $k$ trees, to have a universal tree set (in particular for the minimal case, $k= U(n)$)? 

There are several natural combinatorial questions that can also be asked. For instance, we could define a generalisation of universal incompatibility $U(n,k)$, in which rather than requiring every split in $\mathcal{S}(X)$ to be represented in our tree set, we require only that all splits of size $k$ or less are displayed by a tree in the set (with the present paper of course corresponding to the case $k = \lfloor \frac{n}{2} \rfloor$). One may also consider other important split sets in place of $\mathcal{S}(X)$, such as the set of all splits $A|B$ in which some given subsets $A',B' \subset X$ must be placed in different partitions, that is, either $A' \subset A$ and $B' \subset B$ or $A' \subset B$ and $B' \subset A$.

\begin{figure}[H]
	\centering
	\includegraphics[width=15cm]{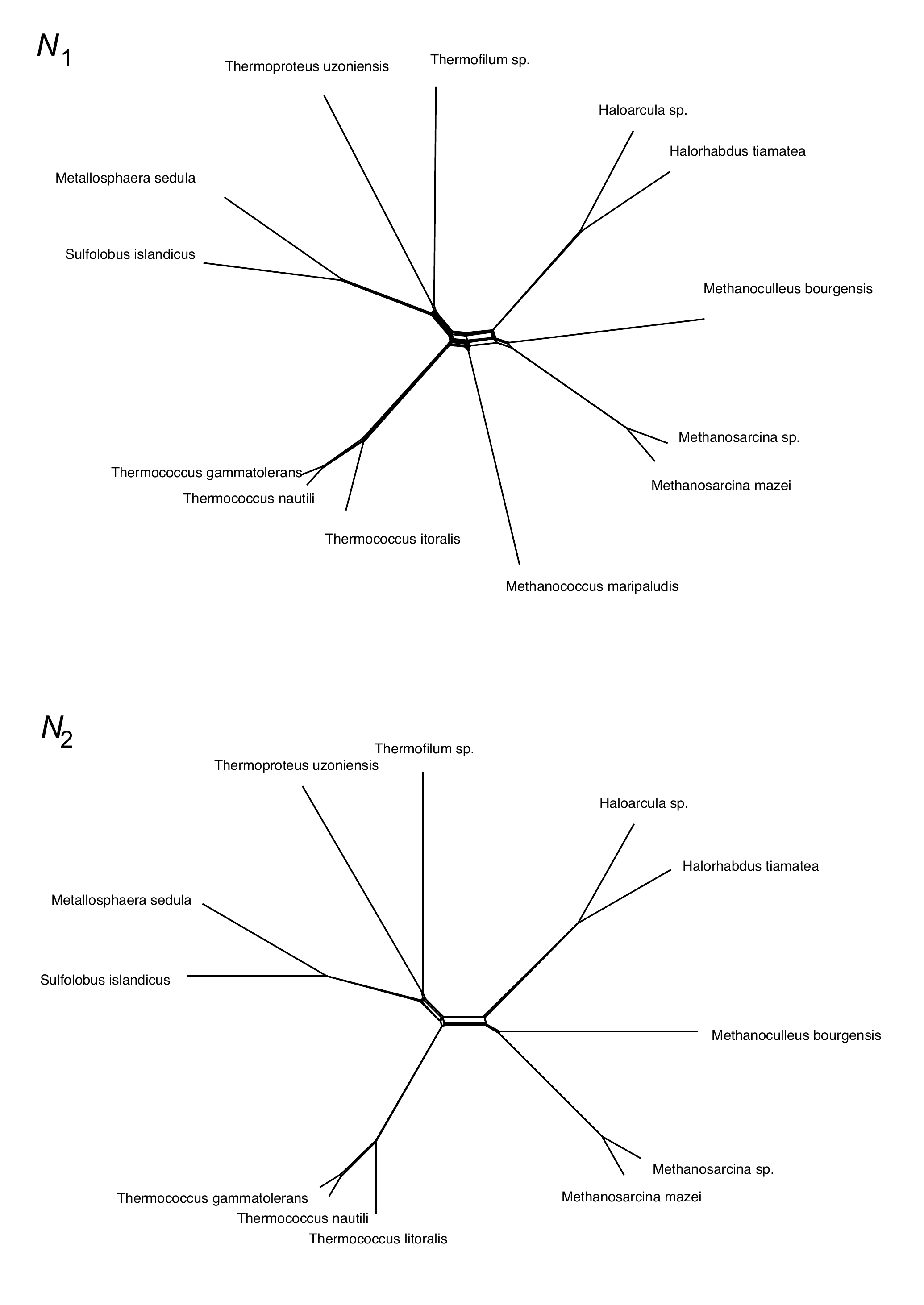}
	\caption{Networks $N_1$ and $N_2$ derived from archaeal genomes.}
	\label{f:archaea}
\end{figure}

\section*{Author Contributions}
MH designed the paper, performed all mathematical research and wrote the majority of the paper. NK prepared and analysed the data, and wrote the Applications and Discussion section.

\bibliographystyle{plainnat}
\bibliography{univ-sets}

\begin{thebibliography}{16}
\providecommand{\natexlab}[1]{#1}
\providecommand{\url}[1]{\texttt{#1}}
\expandafter\ifx\csname urlstyle\endcsname\relax
  \providecommand{\doi}[1]{doi: #1}\else
  \providecommand{\doi}{doi: \begingroup \urlstyle{rm}\Url}\fi

\bibitem[Ban et~al.(2014)Ban, Beckmann, Cate, Dinman, Dragon, Ellis,
  Lafontaine, Lindahl, Liljas, Lipton, et~al.]{ban2014new}
N.~Ban, R.~Beckmann, J.~H.~D. Cate, J.~D. Dinman, F.~Dragon, S.~R. Ellis,
  D.~L.~J. Lafontaine, L.~Lindahl, A.~Liljas, J.~M. Lipton, et~al.
\newblock A new system for naming ribosomal proteins.
\newblock \emph{Current {O}pinion in {S}tructural {B}iology}, 24:\penalty0
  165--169, 2014.

\bibitem[Buneman(1971)]{buneman1971recovery}
P.~Buneman.
\newblock The recovery of trees from measures of dissimilarity.
\newblock \emph{Mathematics in the Archaeological and Historical Sciences},
  1971.

\bibitem[Dilworth(1950)]{dilworth1950decomposition}
R.~P. Dilworth.
\newblock A decomposition theorem for partially ordered sets.
\newblock \emph{Math}, 50:\penalty0 161--166, 1950.

\bibitem[Dress et~al.(2001)Dress, Klucznik, Koolen, and Moulton]{dress20012kn}
A.~Dress, M.~Klucznik, J.~Koolen, and V.~Moulton.
\newblock { $2kn - {2k+1 \choose 2}$}: A note on extremal combinatorics of
  cyclic split systems.
\newblock \emph{S{\'e}minaire Lotharingien de Combinatoire}, 47:\penalty0
  Article B47b, 17, 2001.

\bibitem[Felsenstein(2004)]{felsenstein2004inferring}
J.~Felsenstein.
\newblock \emph{Inferring phylogenies}, volume~2.
\newblock Sinauer associates Sunderland, MA, 2004.

\bibitem[Goyal et~al.(2016)Goyal, Zhou, and Karimi]{goyal2016metabolic}
N.~Goyal, Z.~Zhou, and I.~A. Karimi.
\newblock Metabolic processes of {M}ethanococcus maripaludis and potential
  applications.
\newblock \emph{Microbial {C}ell {F}actories}, 15\penalty0 (1):\penalty0 107,
  2016.

\bibitem[Hendriksen(2020)]{Hendriksen2020}
M.~Hendriksen.
\newblock Minimal{T}ree{S}ets. https://github.com/mahendriksen/minimaltreesets,
  2020.
\newblock URL \url{https://github.com/mahendriksen/MinimalTreeSets}.
\newblock GitHub repository.

\bibitem[Huson and Bryant(2006)]{huson2006application}
D.~H. Huson and D.~Bryant.
\newblock Application of phylogenetic networks in evolutionary studies.
\newblock \emph{Molecular {B}iology and {E}volution}, 23\penalty0 (2):\penalty0
  254--267, 2006.

\bibitem[Karzanov and Lomonosov(1978)]{inbook}
Alexander Karzanov and Michael Lomonosov.
\newblock \emph{Systems of flows in undirected networks}, pages 59--66.
\newblock 01 1978.

\bibitem[Katoh et~al.(2002)Katoh, Misawa, Kuma, and Miyata]{katoh2002mafft}
K.~Katoh, K.~Misawa, K.~Kuma, and T.~Miyata.
\newblock {MAFFT}: a novel method for rapid multiple sequence alignment based
  on fast {F}ourier transform.
\newblock \emph{Nucleic {A}cids {R}esearch}, 30\penalty0 (14):\penalty0
  3059--3066, 2002.

\bibitem[Kummer(1852)]{kummer1852erganzungssatze}
E.~E. Kummer.
\newblock {\"U}ber die {E}rg{\"a}nzungss{\"a}tze zu den allgemeinen
  {R}eciprocit{\"a}tsgesetzen.
\newblock \emph{Journal f{\"u}r die reine und angewandte {M}athematik},
  44:\penalty0 93--146, 1852.

\bibitem[Little et~al.(1975)Little, Grant, and Holton]{little1975defect}
C.~H.~C. Little, D.~D. Grant, and D.~A. Holton.
\newblock On defect-$d$ matchings in graphs.
\newblock \emph{Discrete Mathematics}, 13\penalty0 (1):\penalty0 41--54, 1975.

\bibitem[Nelson-Sathi et~al.(2015)Nelson-Sathi, Sousa, Roettger,
  Lozada-Ch{\'a}vez, Thiergart, Janssen, Bryant, Landan, Sch{\"o}nheit,
  Siebers, et~al.]{nelson2015origins}
S.~Nelson-Sathi, F.~L. Sousa, M.~Roettger, N.~Lozada-Ch{\'a}vez, T.~Thiergart,
  A.~Janssen, D.~Bryant, G.~Landan, P.~Sch{\"o}nheit, B.~Siebers, et~al.
\newblock Origins of major archaeal clades correspond to gene acquisitions from
  bacteria.
\newblock \emph{Nature}, 517\penalty0 (7532):\penalty0 77--80, 2015.

\bibitem[{OEIS Foundation Inc.}(2020)]{OEISseq}
{OEIS Foundation Inc.}
\newblock The {O}n-{L}ine {E}ncyclopedia of {I}nteger {S}equences.
\newblock \url{https://oeis.org/A002661}, 2020.

\bibitem[O'Leary et~al.(2016)O'Leary, Wright, Brister, Ciufo, Haddad, McVeigh,
  Rajput, Robbertse, Smith-White, Ako-Adjei, et~al.]{o2016reference}
N.~A. O'Leary, M.~W. Wright, J.~R. Brister, S.~Ciufo, D.~Haddad, R.~McVeigh,
  B.~Rajput, B.~Robbertse, B.~Smith-White, D.~Ako-Adjei, et~al.
\newblock Reference sequence (refseq) database at ncbi: current status,
  taxonomic expansion, and functional annotation.
\newblock \emph{Nucleic acids research}, 44\penalty0 (D1):\penalty0 D733--D745,
  2016.

\bibitem[Sperner(1928)]{sperner1928satz}
E.~Sperner.
\newblock Ein {S}atz {\"u}ber {U}ntermengen einer endlichen {M}enge.
\newblock \emph{Mathematische Zeitschrift}, 27\penalty0 (1):\penalty0 544--548,
  1928.

\end{thebibliography}
\end{document}